\documentclass[sigconf]{acmart}
\usepackage{amsmath}
\usepackage{amsthm}
\usepackage{thmtools}
\usepackage[ruled,vlined,linesnumbered]{algorithm2e}
\usepackage{color}
\usepackage{verbatim}
\usepackage{cleveref}
\usepackage{enumitem}
\theoremstyle{definition}

\newtheorem{theorem}{Theorem}
\newtheorem{corollary}[theorem]{Corollary}
\newtheorem{lemma}[theorem]{Lemma}
\newtheorem{definition}[theorem]{Definition}

\newlength{\protowidth}
\newcommand{\pprotocol}[5]{
{\begin{figure*}[#4]
\begin{center}
\setlength{\protowidth}{14.9cm}
\addtolength{\protowidth}{-2\intextsep}

\fbox{
        \small
        \hbox{\quad
        \begin{minipage}{\protowidth}
    \begin{center}
    {\bf #1}
    \end{center}
        #5
        \end{minipage}
        \quad}
        }
        \caption{\label{#3} #2}
\end{center}
\vspace{-4ex}
\end{figure*}
} }

\newcommand{\ignore}[1]{}
\newcommand{\pk}{\mathsf{pk}}
\newcommand{\sk}{\mathsf{sk}}
\newcommand{\sign}{\mathsf{Sign}}
\newcommand{\ver}{\mathsf{Verify}}
\newcommand{\sig}[2]{\left<#1\right>_#2}

\newcommand{\HS}{\mathcal{H}}
\newcommand{\SK}{\mathcal{S}}
\newcommand{\minP}{P_\mathsf{min}}
\newcommand{\maxP}{P_\mathsf{max}}
\newcommand{\CB}{\mathsf{CB}}
\newcommand{\TCB}{\mathsf{TCB}}
\newcommand{\CPS}{\mathsf{CPS}}
\newcommand{\APA}{\mathsf{APA}}

\newcommand{\Ex}{\mathsf{Ex}}
\newcommand{\TEx}{\widetilde{\mathsf{Ex}}}

\AtBeginDocument{%
  \providecommand\BibTeX{{%
    \normalfont B\kern-0.5em{\scshape i\kern-0.25em b}\kern-0.8em\TeX}}}

\copyrightyear{2022}
\acmYear{2022}
\setcopyright{rightsretained}
\acmConference[PODC '22]{Proceedings of the 2022 ACM Symposium on Principles of Distributed Computing}{July 25--29, 2022}{Salerno, Italy}
\acmBooktitle{Proceedings of the 2022 ACM Symposium on Principles of Distributed Computing (PODC '22), July 25--29, 2022, Salerno, Italy}\acmDOI{10.1145/3519270.3538444}
\acmISBN{978-1-4503-9262-4/22/07}

\acmConference[PODC '22]{ACM Symposium on Principles of Distributed Computing}{July 25--29, 2022}{Salerno, Italy}



\title{Optimal Clock Synchronization with Signatures}
\titlenote{This project has received funding from the European Research Council (ERC) under the European Union's Horizon 2020 research and innovation programme (grant agreement 716562).}

\begin{document}

\author{Christoph Lenzen}
\email{lenzen@cispa.de}
\affiliation{\institution{CISPA Helmholtz Center for Information Security}
\country{Germany}}
\author{Julian Loss}
\email{loss@cispa.de}
\affiliation{\institution{CISPA Helmholtz Center for Information Security}
\country{Germany}}

\begin{abstract}
Cryptographic signatures can be used to increase the resilience of distributed systems against adversarial attacks, by increasing the number of faulty parties that can be tolerated. While this is well-studied for consensus, it has been underexplored in the context of fault-tolerant clock synchronization, even in fully connected systems. Here, the honest parties of an $n$-node system are required to compute output clocks of small skew (i.e., phase offset) despite local clock rates varying between $1$ and $\vartheta>1$, end-to-end communication delays varying between $d-u$ and $d$, and the interference from malicious parties. Known algorithms with (trivially optimal) resilience of $\lceil n/2\rceil-1$ improve over the tight bound of $\lceil n/3\rceil-1$ holding without signatures for \emph{any} skew bound~\cite{DHS84,ST85}, but incur skew $d$~\cite{ADDNR19} or $\Omega(n(u+(\vartheta-1)d))$~\cite{LM84}. Since typically $d\gg u$ and $\vartheta-1\ll 1$, this is far from the lower bound of $u+(\vartheta-1)d$ that applies even in the fault-free case~\cite{BW01}.

We prove tight bounds of $\Theta(u+(\vartheta-1)d)$ on the skew for the resilience range from $\lceil n/3\rceil$ to $\lceil n/2\rceil-1$. Our algorithm is, granted that the adversary cannot forge signatures, deterministic. Our lower bound holds even if clocks are initially perfectly synchronized, message delays between honest nodes are known, $\vartheta$ is arbitrarily close to one, and the synchronization algorithm is randomized. This has crucial implications for network designers that seek to leverage signatures for providing more robust time. In contrast to the setting without signatures, they must ensure that message delay is at least $d-u$, even on links with one faulty endpoint.
\end{abstract}

\maketitle

\section{Introduction}
Synchronizing clocks is a key task in distributed systems, which has received extensive attention over the years.
Given that distributed systems are prone to faults and attacks, a sizeable fraction of this literature is dedicated to studying \emph{fault-tolerant} clock synchronization.
Under faults, the question of whether clocks can be synchronized has a fundamental impact:
A message passing system prone to crash faults can simulate synchronous execution if and only if communication delay satisfies a known bound and time can be locally approximately measured.
If either assumption fails to hold, the FLP proof of impossibility of deterministic consensus applies~\cite{FLP85}.
Conversely, if both assumptions hold, a network synchronizer~\cite{A85} can be implemented by detecting crash faults via timeout.
At the same time, these conditions are necessary and sufficient to compute logical clocks of bounded \emph{skew,} i.e., bounded maximum difference of concurrent clock readings, on each connected component of the network.

Naturally, assuming crash faults is too optimistic in practice.
However, for any less benign fault model, the situation is similar.
On the one hand, from the above we know that it is necessary to have bounded delay and a local sense of the progress of time to be able to simulate synchronous execution.
On the other hand, these assumptions imply that if we are given a synchronizer, (i) it is guaranteed to complete simulation of a round within bounded time and (ii) we can enforce an arbitrarily large minimum duration of each simulated round (by letting nodes sleep for some time).
Intuitively, this allows us to solve the clock synchronization task by using the current round number as ``target'' logical clock value and from this compute logical clocks of bounded skew and rates by interpolation (this intuition is formalized in~\cite[Ch.~9, Sec.~3.3.4]{Lec21}).

In light of the above, the task of clock synchronization can be seen as a more general and precise version of running a synchronizer~\cite{A85}:
\begin{itemize}
  \item Logical clocks with bounded skew and rates of progress can be readily used to implement a synchronizer.
  The maximum duration of a simulated round then is $r (d+\SK)$, where $r$ is the ratio between maximum and minimum clock rate, $d$ is the (maximum) communication delay, and $\SK$ is (the bound on) the skew.
  \item In contrast to a synchronizer, clocks can also be used to coordinate actions in terms of real time.
\end{itemize}
Note that both of these application scenarios share the property that controlling the rates of the computed logical clocks as well as their skew as precisely as possible matters.

To understand how well clocks can be synchronized under realistic faults or even malicious interference, the research community studied the clock synchronization task in the presence of Byzantine (i.e., worst-case) faults.
In striking similarity to consensus~\cite{LSP82}, it is possible to synchronize clocks in a fully connected message passing system with authenticated channels if and only if the number of Byzantine faults is strictly less than one third~\cite{ST85,DHS84}.
This is good news: one can run synchronous consensus on top of clock synchronization, without any negative impact on resilience!
Even better, this is also true with respect to performance: the same asymptotic bounds on skew can be achieved as in the fault-free case, without any loss in resilience~\cite{BW01,LL84}.

Concretely, denote by $d-u$ the \emph{minimum} communication delay, i.e., the time between a message being sent and the receiving node completing to process it is between $d-u$ and $d$, and let $\vartheta>1$ be the maximum rate of the local reference clocks, whose minimum rate we normalize to $1$.
The main result of~\cite{LL84} can then be read as saying that $\SK \in O(u+(\vartheta-1)d)$ and $r\in 1 + O(\vartheta-1)$ can be achieved, so long as $\vartheta<\vartheta_0$ for a constant $\vartheta_0$ and strictly less than one third of the nodes are faulty, cf.~\cite{KL18}.
This means that using the computed clocks to simulate synchronous execution, each simulated round takes $r(d+\SK)\in d+O(u+(\vartheta-1)d)$ time.
Given that in practice $u\ll d$ and $\vartheta-1 \ll 1$, we have that $d+O(u+(\vartheta-1)d)\approx d$.
In other words, synchrony can be simulated with negligible overhead in time!

For these theoretical results to be of practical value, it is paramount to minimize the overheads incurred by achieving fault-tolerance.
While the resilience bound is tight under the assumption of authenticated channels, it is a well-known result that when \emph{messages} are authenticated, consensus can be achieved in a fully connected system when up to $\lceil n/2\rceil-1$ nodes are faulty~\cite{DS83}.
Indeed, also when synchronizing clocks the resilience can be boosted to $\lceil n/2\rceil-1$ by authenticating broadcasts~\cite{ST85,HSSD84}.
However, these algorithms have skew $\Theta(d)\gg u$. Using signature-based consensus, in~\cite{LM84} optimal resilience is achieved with skew $O(n(u+(\vartheta-1)d))$, where $n$ is the number of nodes; replacing the consensus routine with a faster one could reduce, but not completely eliminate the dependence on $n$. This begs the question
\begin{quote}
\emph{``Which skew can be obtained with signatures at optimal resilience $\lceil n/2\rceil-1$?''}
\end{quote}

\subsection*{Our Contribution}
In this work, we show that the answer to this question is nuanced.
To obtain an asymptotically optimal upper bound on the achievable skew, we set out on the track that readers familiar with the literature might expect.
The algorithm from~\cite{LL84} is based on simulating iterations of synchronous \emph{approximate agreement.} 
\begin{definition}[Approximate Agreement] Let $\Pi$ be protocol executed among $n$ nodes where each node $v$ holds an input $r_v\in\mathbb{R}$ and nodes terminate upon generating an output $o_v\in\mathbb{R}$. Denoting by $\HS$ the set of honest nodes, we say that $\Pi$ is a \emph{$(\ell,\epsilon,f)$-secure protocol for approximate agreement} if the following properties hold whenever at most $f$ nodes are corrupted and $\max_{v,w\in\HS}\{r_v-r_w\}\leq\ell$:
\begin{itemize}
\item $\epsilon$-Consistency: $\max_{v,w\in\HS}\{o_v-o_w\}\leq\epsilon$.
\item Validity: For all $v\in\HS$, $\min_{w\in\HS}\{r_w\}\leq o_v\leq\max_{w\in\HS}\{r_w\}$.
\end{itemize}
\end{definition}
Without signatures, this task can be solved if and only if fewer than one third of the nodes are faulty~\cite{DLPSW86,FLM85}.
The reason is that faulty nodes might claim different inputs to different nodes.
Signatures can overcome this by being able to prove to others which input a sender claimed.
While achieving consensus on inputs might take a non-constant number of rounds even with randomization, for approximate agreement it is sufficient that at most one value from each sender is accepted by correct nodes.
This is easily achieved by outputting ``$\bot$'' -- no value -- if another node proved to receive a conflicting value.
The resulting relaxed variant of reliable broadcast is referred to as \emph{crusader broadcast} in the literature~\cite{D82}.

We show that when communication is by crusader broadcast, approximate agreement can be solved in a logarithmic number of rounds with resilience $\lceil n/2\rceil-1$.
\begin{corollary}\label{cor:apa} There is an $(\ell,\epsilon,\lceil n/2\rceil-1)$-secure protocol for approximate agreement running in $2\lceil\log(\ell/\epsilon)\rceil$ rounds.
\end{corollary}
In contrast to standard approximate agreement, for synchronization purposes the crucial ``content'' of the messages is their \emph{timing.}
With message delays between $d-u$ and $d$, we can simulate crusader broadcast in a timed fashion, which allows each node to estimate the offset between its own and other correct nodes in the system with an error of $O(u+(\vartheta-1)d)$.
The property that faulty nodes communicate no different values then is relaxed, too, in that if two correct nodes accept a broadcast from the same sender, then their respective reception times agree up to $O(u+(\vartheta-1)d)$.

With this subroutine in place, we adapt the algorithm from~\cite{LL84} by replacing plain broadcasts by (simulated) crusader broadcasts and adjusting clocks in accordance with the more resilient approximate agreement algorithm with signatures.
The algorithm solves so-called \emph{pulse synchronization}, which (up to minor order terms) is equivalent to computing logical clocks of small skew and bounded rates at all times, cf.~\cite[Ch.~9, Sec.~3.3.3 and 3.3.4]{Lec21}.
\begin{definition}[Pulse Synchronization] Let $\Pi$ be a protocol executed among $n$ nodes, where the set of honest nodes is $\HS$. We say that $\Pi$ is an \emph{$f$-secure protocol for pulse synchronization with skew $\SK$, minimum period $\minP>0$, and maximum period $\maxP$} if the following properties hold whenever at most $f$ nodes are corrupted:
\begin{itemize}
\item Liveness: for all $i\in\mathbb{N}_{>0}$ and each node $v\in\HS$, $v$ outputs pulse $i$ exactly once. We denote $p_{v,i}$ as the time where $v$ outputs its $i$th pulse.
\item $\SK$-bounded skew: $\sup_{i\in\mathbb{N}_{>0},v,w\in\HS}\{|p_{v,i}-p_{w,i}|\}\leq\SK$
\item $\minP$-minimum period: $$\inf_{i\in\mathbb{N}}\{\min_{v\in\HS}\{p_{v,i+1}\}-\max_{v\in\HS}\{p_{v,i}\}\}\geq\minP.$$
\item $\maxP$-maximum period: $$\sup_{i\in\mathbb{N}_{>0}}\{\max_{v\in\HS}\{p_{v,i+1}\}-\min_{v\in\HS}\{p_{v,i}\}\}\leq\maxP.$$
\end{itemize}
In case of a randomized algorithm, the skew bound is allowed to depend on the randomness of the algorithm; $\SK$ is then defined as the expected worst-case value, where the expectation is taken over the randomness of the algorithm.
\end{definition}
\begin{restatable}{corollary}{sync}\label{cor:sync} If $\vartheta\le 1.11$, there are choices $T\in \Theta(d)$ and $S\in \Theta(u+(\vartheta-1)d)$ such that Algorithm~\ref{fig:cps} is a $(\lceil n/2\rceil-1)$-secure pulse synchronization protocol with skew $S$, minimum period $P_{\min}\in \Theta(d)$ and maximum period $P_{\max}\in P_{\min}+\Theta(S)$.
\end{restatable}

Since this skew bound is asymptotically optimal even without faults~\cite{BW01}, at first glance it might appear that this settles our above question.
However, there is a crucial difference to the signature-free setting.
For the above result, it is necessary that also \emph{faulty} nodes must obey the minimum message delay of $d-u$, both when receiving and sending messages.
Otherwise, they could obtain and send a signature used by a correct sender in a crusader broadcast so early that correct nodes reject the sender's broadcast.
This is in stark contrast to the algorithm from~\cite{LL84}, for which faulty nodes can have full information of the system state at all times, including the future!

Given that it might be very challenging or even impossible for system designers to guarantee that an attacker must obey a minimum communication delay of $d-u$ for $u\ll d$, we need to determine whether this limitation is inherent.
Perhaps surprisingly, we prove that this is indeed the case, by providing a matching lower bound.
If either messages to or from faulty nodes have delays from $[d-\tilde{u},d]$ for some $\tilde{u}\in [u,d]$, then we can prove a lower bound of $\Omega(\tilde{u})$ on the skew, regardless of $u$.
\begin{restatable}{theorem}{lb}\label{thm:lb} Let $n\geq 3$ and $\Pi$ be an $\lceil n/3\rceil$-secure protocol for pulse synchronization with skew $S$. Then $\mathbb{E}[S]\geq 2\tilde{u}/3$.
\end{restatable}
Since a lower bound of $(\vartheta-1)d$ follows from a simple indistinguishability argument even in absence of faults, we hence establish an asymptotically tight bound of $\Theta(\tilde{u}+(\vartheta-1)d)$ on the skew that can be achieved when the number of faults is at least $n/3$. We stress that \Cref{thm:lb} imposes no restrictions on $\Pi$, which might be randomized, holds under perfect initial synchronization, for arbitrarily small $\vartheta-1$ and $u=0$, and our adversary is static, i.e., chooses which nodes to corrupt upfront. In other words, all the typical loopholes one might try to exploit to improve on our upper bound result are unavailable, implying asymptotic optimality of its skew in a strong sense.

\paragraph*{Organization of this article.}
In~\Cref{sec:prelim}, we specify our model and cover some preliminaries, including synchronous approximate agreement with resilience $\lceil n/2\rceil-1$. In~\Cref{sec:cps}, we provide and analyze our Crusader Pulse Synchronization algorithm. \Cref{sec:lower} proves the lower bound. Due to space constraints, discussion of further related work, some proofs, and a synchronous Crusader Broadcast algorithm are deferred to Appendices~\ref{app:further}, \ref{app:proofs}, and Figure~\ref{fig:cb}, respectively.

\section{Preliminaries and Model}\label{sec:prelim}
We consider a network of $n$ nodes connected by pairwise, authenticated channels. An unknown subset of the nodes is faulty or even malicious; we denote by $\HS$ the set of the remaining honest nodes. We also assume that nodes have established a \emph{public key infrastructure (PKI)}. This means that every node $v$ has a \emph{public key $\pk_v$} that all other nodes agree on. Any honest node is also assumed to hold a matching \emph{secret key $\sk_v$} with which it can create a \emph{signature $\sig{m}{v}$} on a message $m$ via $\sig{m}{v}\gets\sign(\sk_v,m)$. 
A signature can be verified via $\ver(\pk_v,\sig{m}{v},m)$, which returns a bit 0 \emph{(invalid)} or 1 \emph{(valid)}.
Therefore, we assume that a signature with respect to $\pk_v$ on any message $m$ is impossible to create without knowledge of $\sk_v$. We also assume \emph{perfect correctness}: for any message $m\in\{0,1\}$, $\ver(\pk_v,\sign(\sk_v,m),m)=1$.

\medskip\noindent\textbf{Assumptions on the Network.} We use a continuous notion of time; thus time takes values in $\mathbb{R}_{\ge 0}$. We assume a fully connected network with known minimum and maximum delays of $d-u$ and $d$, respectively. This means that any message sent to or from an honest node is delivered after at most $d$ time and no faster than $d-u$ time, where $d$ and $u$ are known parameters. We refer to $u$ as the \emph{uncertainty}. When proving our lower bound, we will allow for the possibility that messages to and from faulty nodes might violate the minimum delay bound of $d-u$ and instead satisfy only a weaker bound of $d-\tilde{u}$ for some $\tilde{u}\in [u,d]$. Note that we assume that the network is fault-free, since we are free to map link failures to node failures. Finally, we assume that only finitely many messages are sent in finite time; this must clearly be satisfied in any real-world network and it simplifies our lower bound construction by enabling us to perform induction over the messages sent by the honest parties executing an arbitrary, but fixed algorithm.

\medskip\noindent\textbf{Hardware Clocks and Clock Rate.} Nodes have no access to the ``true'' time $t\in \mathbb{R}_{\ge 0}$. Instead, each node $v$ can measure the progress of time approximately via its \emph{hardware clock.} The hardware clock of node $v$ is modelled by a function $H_v\colon \mathbb{R}_{\ge 0}\longrightarrow \mathbb{R}_{\ge 0}$ which the node can evaluate at any point in time. That is, $H_v$ maps time $t\in\mathbb{R}_{\ge 0}$ to a \emph{local time} $H_v(t)\in\mathbb{R}_{\ge 0}$. We assume that hardware clocks run at rates between $1$ and $\vartheta$ for a known constant $\vartheta>1$, i.e., for all $t'\ge t\in \mathbb{R}_{\ge 0}$, it holds that
\begin{equation*}
t'-t\le H_v(t')-H_v(t)\le \vartheta(t'-t).
\end{equation*}
Furthermore, we assume some degree of initial synchrony, represented by bounding $\max_{v,w\in \HS}\{|H_v(0)-H_w(0)|\}$.
For simplicity, we assume this bound to equal $S$ when showing our upper bound; our lower bound result holds under the assumption of perfect initial synchrony, i.e., $H_v(0)=H_w(0)$ for all $v,w\in \HS$.

\medskip\noindent\textbf{Adversary and Executions.} The adversary is in full control of message delays and hardware clocks within the bounds specified by the model. That is, the adversary specifies arbitrary hardware clock functions and message delays subject to the above constraints. Our algorithm is resilient to any adversary that cannot forge signatures, i.e., it is deterministically correct under this assumption. On the other hand, our lower bound holds for a static adversary that decides which nodes to corrupt upfront. The adversary may use corrupted nodes' secrets to generate signatures for them, but needs to obtain signatures of honest nodes affecting a message it intends to send before it can generate the message. Apart from this restriction, the adversary fully controls the behavior of faulty nodes.

Formally, an \emph{execution} is fully specified by determining $\HS$, $H_v$ for each $v\in \HS$, which messages faulty nodes send and when, and the delays of all messages. An execution is \emph{well-defined}, i.e., conforms with our model, if these parameters specify the above specification and for each message $m$ sent by a faulty node at time $t$, it receives all messages $m'$ containing signatures from honest nodes that $m$ depends on, some faulty node received a message $m''$ with the same signature by time $t$.\footnote{In our lower bound construction, we simply guarantee that the faulty sender of $m$ receives $m'$ by time $t$.}

\medskip\noindent\textbf{Synchronous Execution and Rushing Adversary.}
For the sake of clarity of presentation, we use a classic synchronous model in our description and analysis of our solution to the approximate agreement problem.
That is, computation proceeds in compute-send-receive rounds, and the goal is to complete the task in as few communication rounds as possible.
While working in the synchronous model, we also assume a rushing adversary that can immediately observe honest nodes' messages in any given round and choose its own messages based on them.
However, since our goal is to synchronize clocks in the above model, we will later show how to overcome this assumption in our main synchronization protocol. 

\medskip\noindent\textbf{Crusader Broadcast.}
In the synchronous setting, we assume a subroutine implementing Crusader Broadcast to be available.

\begin{definition}[Crusader Broadcast] Let $\Pi$ be protocol executed among $n$ nodes where a designated dealer $v$ holds an input $b_v\in\{0,1\}$ and nodes terminate upon generating an output $o_v\in\{0,1,\bot\}$. We say that $\Pi$ is an \emph{$f$-secure protocol for crusader agreement} if the following holds whenever at most $f$ nodes are corrupted:
\begin{itemize}
\item Validity: If $v\in\HS$, then $o_w=b_v$ for all $w\in\HS$.
\item Crusader Consistency: If $o_u\in\{0,1\}$ for some $u\in\HS$, then for all $w\in\HS$, $o_w\in\{\bot,o_u\}$.
\end{itemize}
\end{definition}
Figure~\ref{fig:cb} shows Algorithm~$\CB$ implementing crusader broadcast with signatures, whose correctness is shown in~\cite{Dolev82}.

\subsection{Approximate Agreement with Signatures}
In this section, we describe a simple synchronous $2$-round $(\ell,\ell/2,\allowbreak\lceil n/2\rceil-1)$-secure approximate agreement algorithm, where we describe the protocol from the view of node $v$ with input $r_v$ and $\ell$ denotes the initial range of honest nodes' values. Note that a $(\ell,\epsilon,\lceil n/2\rceil-1)$-secure approximate agreement algorithm follows immediately by repeating the algorithm $\lceil\log \ell/\epsilon\rceil$ times, feeding the output of the previous iteration as input into the new instance.

The algorithm is given in~\Cref{fig:aa}. Intuitively, its correctness follows from the facts that (i) nodes always discard enough values to ensure that the remaining ones lie within input range, where any received $\bot$ guarantees that the respective sender is faulty and does not contribute to the list of received non-$\bot$ values and (ii) the intervals spanned by the remaining values must intersect, implying that their midpoints can only be $\ell/2$ apart.

\pprotocol{Algorithm $\APA$}{A synchronous $2$-round $(\ell,\ell/2,f)$-secure algorithm for approximate agreement, where $f=\lceil n/2\rceil-1$.}{fig:aa}{}{
\begin{itemize}
\item $v$ sends $r_v$ to all nodes using (an instance of) Algorithm~$\CB$
\item Denote $r_w$ the value received from $w$ (which might be $\bot$). Determine $o_v$ as follows:
\begin{itemize}
\item Denote $b$ as the number of $\CB$ instances that output $\bot$ in the previous step.
\item Sort all the non-$\bot$ values received via the instances of $\CB$ in the previous step, then discard the lowest $f-b$ and the highest $f-b$ of those values. Denote $I$ the interval spanned by the remaining values.
\item Output the midpoint $o_v$ of $I$.
\end{itemize}
\end{itemize}
}

The first point is immediate from the properties of crusader broadcast. To show the second point, we argue as follows. First, observe that if there are no $\bot$-values received by honest parties, all of them have the same lists and, trivially, retain the same interval. Second, a node receiving a $\bot$-value can only increase the interval it retains, allowing us to extend this statement by induction on the number of $\bot$ values received by honest parties to all executions.

We now formalize this intuition. To this end, we denote as $I_v:=[a_v,b_v]$ the interval spanned by the remaining values received by node $v\in\HS$ in the first step  (i.e., the ones that remain after discarding the highest and lowest values for that iteration). For any execution of the protocol in which some $v\in\HS$ receives at least one $\bot$ value and any $x\in\mathbb{R}$, we denote by $\tilde{I}^{x}_v=[\tilde{a}_v,\tilde{b}_v]$ the interval $I_v$ that results from an alternative execution of the protocol in which $v$ receives identical messages except that one of the $\bot$ values is replaced by $x$.

\begin{lemma}\label{lemma:removebots} For all $x\in\mathbb{R}$ and all $v\in\HS$, $\tilde{I}^{x}_v\subseteq I_v$.
\end{lemma}
\begin{proof} Fix an arbitrary node $v\in\HS$ receiving at least one $\bot$ value. For $0< b\le f$, denote by $r^1_v,...,r^{n-k}_v$ the $n-b$ non-$\bot$ values that $v$ receives in this iteration and assume they are ordered in ascending fashion. Similarly, denote by $\tilde{r}^1_v,...,\tilde{r}^{n-k+1}_v$ the corresponding values in the alternative execution in which in addition $x$ is received by $v$. Because the second list is identical to the old one except for exactly one additional value, we have for all $1\le j\le n-b$ that $\tilde{r}^{j}_v \le r^{j}_v\le \tilde{r}^{j+1}_v$. In particular,
\begin{equation*}
\tilde{a}_v = \tilde{r}_v^{f-(b-1)+1} \ge r_v^{f-b+1} =a_v
\end{equation*}
and
\begin{equation*}
\tilde{b}_v = \tilde{r}_v^{n-(f-(b-1))} \le r_v^{n-(f-(b-1))} \le r_v^{n-(f-b)}=b_v,
\end{equation*}
i.e., $\tilde{I}^x_v=[\tilde{a}_v,\tilde{b}_v]\subseteq [a_v,b_v]=I_v$.
\end{proof}

\begin{lemma}\label{lemma:median} There exists $x\in \mathbb{R}$ such that $x\in I_v$ for all $v\in\HS$.
\end{lemma}
\begin{proof} We prove the claim by induction over the total number $k$ of $\bot$ values received by honest nodes. 
\begin{itemize}
\item \textbf{base case:} Suppose that $k=0$. Thus, for all $v,w\in\HS$, $I_v=I_w$ by crusader consistency of $\CB$. Since $f=\lceil n/2\rceil-1$, each node retains at least one value, i.e., this interval is non-empty.
\item \textbf{step from $k$ to $k+1$:} Suppose the claim holds for $k\in \mathbb{N}$. By the induction hypothesis, there is $x\in \mathbb{R}$ such that $x\in I_v$ for all $v\in \HS$. Consider an alternative execution in which all nodes receive the same messages except that for some honest node $w$, one of its values is replaced by $\bot$. Denote by $\tilde{I}_v$, $v\in \HS$, the intervals in this new execution. For each honest node $v\neq w$, we have that $x\in I_v=\tilde{I}_v$. For $w$, we have that $x\in I_w\supseteq \tilde{I}_w$ by \Cref{lemma:removebots}.
\end{itemize}
Since this construction applies for any inputs $r_v$, $v\in \HS$, and we imposed no restrictions on the values received from faulty nodes, the claim of the lemma follows for all executions of the algorithm.
\end{proof}

\begin{theorem}\label{thm:apa} Algorithm~\ref{fig:aa} is an $(\ell,\ell/2,\lceil n/2\rceil-1)$-secure $2$-round protocol for approximate agreement.
\end{theorem}
\begin{proof}
Without loss of generality, we order the inputs of the honest nodes in ascending order. As there are at most $f=\lceil n/2\rceil-1$ corrupted nodes and $\CB$ has validity, we have that $r_1\leq a_v$ and $r_{n-f}\geq b_v$ for all $v\in\HS$. Hence, $o_v=(a_v+b_v)/2$ satisfies the validity condition of approximate agreement.

By~\Cref{lemma:median}, there is $x\in \mathbb{R}$ such that $x\in I_v=[a_v,b_v]$ for all $v\in \HS$. For all $v,w\in \HS$, we can hence infer that
\begin{align*}
o_v-o_w&= \frac{a_v+b_v-a_w-b_w}{2}\\
&\leq \frac{x+r_{n-f}-r_1-x}{2}=\frac{r_{n-f}-r_1}{2}\le\frac{\ell}{2}.\qedhere
\end{align*}
\end{proof}
\Cref{cor:apa} readily follos by inductive application of \Cref{thm:apa}, where the output of the $i$th run is used as input for the $(i+1)$th~run.

\section{Crusader Pulse Synchronization}\label{sec:cps}

In this section, we provide and prove correct our pulse synchronization algorithm. In analogy to the classic algorithm by Lynch and Welch~\cite{LL84} that achieves asymptotically optimal skew in the signature-free setting, the algorithm can be viewed as simulating iterations of synchronous approximate agreement. The goal is to agree on the pulse times, i.e., pull them closer together despite interference from faulty nodes.

There are two notable differences from running ``plain'' approximate agreement. One is that the goal is to agree on (real) times, which cannot directly be accessed by the nodes. This results in the algorithm using communication and the hardware clocks to estimate the \emph{differences} between pulse times. Because these measurements are inexact and clocks drift, the second difference emerges: while approximate agreement works towards decreasing the error in each iteration, (in the worst case) these inaccuracies in the nodes' perception work towards increasing them.

In the following, assume parties are running an algorithm for pulse synchronization, where $\vec{p}_r$ denotes the pulse times of honest parties for the $r$th pulse.\footnote{Since a priori there is no formal guarantee that $v\in \HS$ will indeed generate all pulses, one can define $p_v^r:=\infty$ if $v$ generates fewer than $r$ pulses. However, our inductive proof does not need to reason about $p_{v,r+1}$ before it is established that it is finite.} We define $\|\vec{p}_r\|:=\max_{v\in \HS}\{p_v^r\}-\min_{v\in \HS}\{p_v^r\}$. Denote by $S$ the upper bound on $\|\vec{p}_r\|$ for all $r\in\mathbb{N}$ that we are going to show. Note that the algorithm is allowed to make use of $S$, even though we will be able to determine $S$ only once our analysis is complete.

\subsection{Timed Crusader Broadcast}
Before presenting the full algorithm, let us discuss how estimates of clock offsets can be obtained that comply with the requirements of (our simulation of) Algorithm~$\APA$. Intuitively, instead of communicating the (unknown) pulse times, each node $v\in \HS$ will broadcast (up to a small, fixed local time offset) when locally generating its pulse. Knowing that these messages are underway for about $d$ time, each recipient $w\in \HS$ can then determine an approximation of $\Delta_{v,w}^r$ of $p_w^r-p_v^r$ of error $O(u+(\vartheta-1)d)$. 

In order for these estimates to be used in the simulation of Algorithm~$\APA$, we need them to satisfy a timed analogon of Crusader Consistency. In the absence of any error, this would mean that for any two honest parties $v,w\in \HS$ and (possibly faulty) node $x\in \HS$ such that $\Delta_{v,x}^r,\Delta_{w,x}^r\neq \bot$, it holds that $\Delta_{v,x}^r-\Delta_{w,x}^r=p_w^r-p_v^r$. Under this condition, the estimates $p_v^r+\Delta_{v,x}^r$ and $p_w^r+\Delta_{w,x}^r$ are equal, i.e., it is as if faulty node $x$ had broadcast to a subset of the honest nodes at some specific time $p_x^r$ that each $v\in \HS$ with $\Delta_{v,x}^r\neq \bot$ agrees on. Naturally, we cannnot ensure such an exact match, but signatures give us the possibility to prove reception time up to an error of $O(u+(\vartheta-1)d)$. By rejecting the broadcast message from node $x$ if another node proves to have received the same ``broadcast'' more than $u$ (real) time earlier, we can hence ensure the above consistency condition up to an error of $O(u+(\vartheta-1)d)$. The protocol ensuring these guarantees is given in~\Cref{fig:tcb}.

\pprotocol{Algorithm $\TCB^r$}{Timed Crusader Broadcast. The routine assumes that $\|\vec{p}_r\|\le S$, where $p^r_v$ is the time at which $v$ generates the $r$th pulse. Encoding $r\in \mathbb{N}$ allows to distinguish instances, so that faulty nodes cannot reuse ``old'' signatures to disrupt an instance.}{fig:tcb}{t}{
We describe the protocol from the view of node $v$, where $w$ is the dealer. 
\begin{itemize}
\item If $v$ is the dealer, i.e., $v=w$, $v$ sends $\sig{r}{v}$ to all nodes at local time $H_v(p^r_v)+\vartheta\cdot S$ and terminates.
\item If $v$ is not the dealer, i.e., $v\neq w$:
\begin{itemize}
  \item If no correctly formed signature $\sig{r}{w}$ is received from $w$ at a local time $h\in(H_v(p^r_v),H_v(p^r_v)+\vartheta(d+(\vartheta+1)S))$, terminate with output $\bot$.
  \item Otherwise, denote by $h$ the first such local time. Forward $\sig{r}{w}$ to all nodes at time $h$.
  \item If $v$ receives a correctly formed signature $\sig{r}{w}'$ at a local time $h'\in (H_v(p^r_v),h+d-2u)$ from a node $x\neq w$, it terminates with output $\bot$.
  \item Otherwise, $v$ terminates with output $h$ at local time $h+d-2u$.
\end{itemize}
\end{itemize}
}

We now formalize the above claims. We first establish that honest dealers' messages are always accepted, i.e., recipients will output their local reception time, which corresponds to validity of Crusader Broadcast.
\begin{lemma} Let $r\in\mathbb{N}$ and suppose that $\|\vec{p}^r\|\leq S$. Then for all $v,w\in\HS$, $v$ outputs $h\neq\bot$ in $\TCB^r$ with $w$ as the sender.
\label{lem:correctscb1}
\end{lemma}
\begin{proof} We begin by proving that $w$'s signature is received by $v$ at a local time $h\in (H_v(p^r_v),H_v(p^r_v)+\vartheta(d+(\vartheta+1) S))$. Denote by $t_w$ the (unique) time at which $w$ sends $\sig{r}{w}$ and let $t_v\in [t_w+d-u,t_w+d]$ be the time at which $v$ receives it. We observe that
\begin{align*}
S= \frac{H_w(t_w)-H_w(p_w^r)}{\vartheta}\le t_w-p_w^r\le H_w(t_w)-H_w(p_w^r)=\vartheta S.
\end{align*}
Hence,
\begin{align*}
H_v(t_v)&\ge H_v(t_w)\\
&= H_v(t_w)-H_v(p_v^r)+H_v(p_v^r)\\
&\geq t_w-p_v^r+H_v(p_v^r)\\
&\geq p_w^r-p_v^r+S+H_v(p_v^r)\ge H_v(p_v^r),
\end{align*}
where the last step uses that $\|\vec{p}^r\|\le S$ by assumption. For the upper bound on $h$, we get that
\begin{align*}
H_v(t_v)&\le H_v(t_w+d)\\
&\le H_v(p_w^r+\vartheta S+d)-H_v(p_v^r)+H_v(p_v^r)\\
&\leq \vartheta(p_w^r-p_v^r+\vartheta\cdot S+d)+H_v(p_v^r)\\
&\leq \vartheta ((\vartheta+1) S+ d)+ H_v(p_v^r).
\end{align*}
It remains to prove that $v$ receives no correct signature of the form $\sig{r}{w}'$ at a local time $h'\in (H_v(p_v^r),h+d-2u)$ from $x\neq w$. Since $w$ sends no other signature than $\sig{r}{w}$ and signatures cannot be forged, any such message must be sent after $x$ learned $\sig{r}{w}$, at time $t_x\ge t_w+d-u$. Hence, $v$ receives any such message at a time $t_v'\ge t_x+d-u\ge t_w+2(d-u)$. On the other hand $t_v\le t_w+d$, implying that $t_v'-t_v\ge d-2u$. We conclude that $h'-h=H_v(t_v')-H_v(t_v)\ge t_v'-t_v\ge d-2u$.
\end{proof}

Having established that honest nodes broadcasts are accepted, let us establish the counterpart of Crusader Consistency, namely that if honest nodes accept a broadcast, they do so within a short time of each other.
\begin{lemma}\label{lemma:rtbound} Let $r\in\mathbb{N}$ and $u,v\in\HS$. Suppose that $u,v$ participate in an instance of $\TCB^r$ with dealer $w$ and output $h_u,h_v\not\in\{\bot\}$, respectively. Let $t_u,t_v$ denote the times at which $u$ and $v$ receive the messages from $w$, respectively. Then $|t_u-t_v|\leq (1-1/\vartheta)\cdot d+2u/\vartheta$.
\end{lemma}
\begin{proof} Denote as $t_u,t_v$ the times that $u$ and $v$ receive a message from $w$ in $\TCB^r$ s.t. $H_u(t_u)\in(H_u(p^r_u),H_u(p^r_u)+\vartheta(d+\vartheta\cdot S))$ (and similarly for $v$). If for either $u$ or $v$ no such time exists, then that party outputs $\bot$ and the claim is vacuously true.

Without loss of generality, assume that $t_u\geq t_v$. $u$ receives the echo from $v$ at time $t\leq t_v+d$. Since $u$ does not output $\bot$, we have that $\vartheta\cdot(t-t_u)\geq H_u(t)-H_u(t_u)\geq d-2u.$ This implies that $t-t_u\geq (d-2u)/\vartheta$. It follows $t_u-t_v=t-t_v-(t-t_u)\leq d -(d-2u)/\vartheta=(1-1/\vartheta)\cdot d+2u/\vartheta$.
\end{proof}

\subsection{Pulse Synchronization Algorithm}
With Timed Crusader Broadcast in place, we are ready to proceed to the main algorithm, i.e., iterative simulation of synchronous approximate agreement steps on pulse times. The algorithm is given in~\Cref{fig:cps}.

\pprotocol{Algorithm $\CPS$}{Crusader Pulse Synchronization, where $f=\lceil n/2\rceil-1$ is the number of faulty nodes that can be sustained. The algorithm assumes that $H_v(0)\in [0,S]$ for all $v\in \HS$.}{fig:cps}{}{
We describe the protocol from the view of node $v$. Wait until local time $S$.\bigskip

Then do for all $r\in\mathbb{N}$:
\begin{itemize}
\item Generate $r$th pulse.
\item Simultaneously participate in an execution of $\TCB^r$ with sender $w$, for each node $w\in[n]$. Let $h_{v,w}$ denote the output of the instance corresponding to sender $w$.
\item For each node $w\in[n]$ s.t. $h_{v,w}\neq\bot$, compute $\Delta_{v,w}^r:=h_{v,w}-H_v(p_v^r)-d+u-S$. For all remaining $w\in[n]$, set $\Delta_{v,w}^r:=\bot$. Denote $b$ the number of $\bot$ values.
\item Sort all the non-$\bot$ values computed in the previous step, then discard the lowest $f-b$ and the highest $f-b$ of those values. Denote $I$ the interval spanned by the remaining values.
\item Set $\Delta_v^r$ to the midpoint of $I$.
\item Wait until local time $H_v(p_v^r)+\Delta_v^r+T$.
\end{itemize}
}
Intuitively, in each iteration $r\in \mathbb{N}$ the algorithm lets each node $w$ communicate its pulse time using crusader broadcast. The output of this subroutine at $v$ is then used to compute an estimate $\Delta_{v,w}^r$ of $p_w^r-p_v^r$, which in case of a faulty sender might fail and result in $\bot$. These estimates are then used exactly as in Algorithm~$\APA$ to determine a value of $\Delta$, which is then used as a correction to the next pulse time relative to the nominal duration of an iteration of $T$. The fact that Algorithm~$\APA$ is used on the \emph{differences} of pulse times is compensated for by adding its output to the previous local pulse time of $H_v(p_v^r)$; since substracting $p_v^r$ does not change the order of the received values and the output is a convex combination of two inputs of a specific rank, up to the errors introduced by clock drift and delay uncertainty, this is equivalent to executing Algorithm~$\APA$ on inputs $p_v^r$.

We begin our analysis of the algorithm by translating the validity and consistency guarantees of Algorithm~$\TCB^r$ into corresponding guarantees on the computed estimates $\Delta_{v,w}^r$. In the following, we will bound the error in the estimates by $\delta:=2u+(\vartheta^2-1)d+2(\vartheta^3-\vartheta^2)S$. First, we show validity, i.e., that honest nodes compute estimates of their difference in pulse times with error smaller than $\delta$. The result is shown analogously to~\cite[Ch.~10, Lem.~10.8]{Lec21}; we defer the full proof to \Cref{app:proofs}.
\begin{lemma}\label{lemma:deltabh} Let $r\in\mathbb{N}$ and suppose that $\|\vec{p}^r\|\leq S$. Let $v,w\in\HS$ and suppose that $v$ participates in an instance of $\TCB^r$ with dealer $w$. Consider $\Delta^r_{v,w}$ as defined in algorithm $\CPS$. Then $\Delta_{v,w}^r\in[p^r_w-p^r_v,p^r_w-p^r_v+\delta).$
\label{lem:correctscb2}
\end{lemma}
\begin{proof}[Proof Sketch.]
The computation of $\Delta^r_{v,w}$ accounts for the minimum time (and hence local time) that passes before $v$ receives the message from $w$, which is accepts by~\Cref{lem:correctscb1}. The upper bound follows by checking the maximum local time that passes, which is proportional to $\vartheta$ times the real difference between the pulse times plus a delay uncertainty.
\end{proof}
Next, we show the counterpart of Crusader Consistency, i.e., that non-$\bot$ estimates are consistent up to error $\delta$.
\begin{lemma}\label{lemma:deltabf} Let $r\in\mathbb{N}$ and suppose that $\|\vec{p}^r\|\leq S$. Moreover, let $v,w\in\HS$, $x\in[n]/\HS$, and $h_v,h_w\not\in\{\bot\}$ denote the outputs of $v$ and $w$ in $\TCB^r$ with sender $x$. Then $|\Delta_{v,x}^r-\Delta_{w,x}^r-(p^r_w-p^r_v)|<\delta$.
\label{lem:scb1}
\end{lemma}
\begin{proof}
Let again $t_v$ and $t_w$ denote the times at which $v$ and $w$ receive the messages from the dealer $x$. By \Cref{lemma:rtbound}, it holds that $t_v-t_w\le (1-1/\vartheta)d+2u/\vartheta$. We get that
\begin{align*}
&\Delta_{v,x}^r-\Delta_{w,x}^r-(p_v^r-p^r_u)\\
=\,&H_v(t_v)-H_v(p_v^r)-(H_w(t_w)-H_w(p^r_w))-(p_w^r-p^r_v)\\
\leq\,& \vartheta(t_v-p_v^r)-(t_w-p^r_w)-(p_v^r-p^r_w)\\
=\,& (\vartheta-1)(t_v-p_v^r)+t_v-t_w\\
\leq \,&(\vartheta-1)(t_v-p_v^r)+\left(1-\frac{1}{\vartheta}\right) d+\frac{2u}{\vartheta}\\
\leq \,& (\vartheta-1)(H_v(t_v)-H_v(p_v^r))+\left(1-\frac{1}{\vartheta}\right) d+\frac{2u}{\vartheta}\\
\leq \,&(\vartheta-1)(\vartheta d +(\vartheta^2+\vartheta) S)+\left(1-\frac{1}{\vartheta}\right) d+\frac{2u}{\vartheta}<\delta.\qedhere
\end{align*}
\end{proof}

Based on the above bounds, we will prove by induction that for all $r\in \mathbb{N}$, $\|\vec{p}^r\|\le S$. The base of the induction is given by the assumption that hardware clocks are initialized with skew $S$. The following lemma is the key argument required for the step, establishing that, essentially, an approximate agreement step on the pulse times with error at most $\delta$ is performed, assuming that the estimates satisfy the above validity and consistency conditions.
\begin{lemma}\label{lemma:apa_sim} Fix $r\in \mathbb{N}$. Suppose that each $v\in\HS$ computes for each $w\in[n]$ a value $\Delta^r_{v,w}\in\mathbb{R}\cup\{\bot\}$, such that the following properties hold:
\begin{itemize}
\item For $v,w\in\HS$, $\Delta_{v,w}^r\in[p_w^r-p_v^r,p_w^r-p_v^r+\delta]$.
\item For $v,w\in\HS$ and $x\in [n]$ such that $\Delta_{v,x}^r,\Delta_{w,x}^r\not\in\{\bot\}$, $|\Delta^r_{v,x}-\Delta_{w,x}^r-(p_w^r-p_v^r)|\leq \delta$.
\end{itemize}
Then the following statements are true:
\begin{enumerate}
\item For all $v\in\HS$: $-\|\vec{p}^r\|\leq \Delta_v^r\leq \|\vec{p}^r\|+\delta.$
\item $\|\vec{\Delta}^r+\vec{p}^r\|\leq\|\vec{p}^r\|/2+\delta.$
\end{enumerate}
\end{lemma}
\begin{proof} We first show the claim for the special case $\delta=0$. Consider an honest node $v\in\HS$ and a faulty node $x\in[n]\setminus\HS$ such that $\Delta_{v,x}^r\neq\bot$. We define $p^r_x:=\Delta^r_{v,x}+p^r_v$. We observe that $p_v^r+\Delta^r_{v,x}=p_w^r+\Delta^r_{w,x}$ for any correct node $w\in\HS$ from the second condition of the statement and the assumption that $\delta=0$. Hence $p^r_x=p_w^r+\Delta^r_{w,x}$ for all $w\in \HS$ satisfying $\Delta_{w,x}^r\neq\bot$. 

We next show that $p_v^r+\Delta^r_{v}$ for $v\in\HS$ equals the output of $v$ in an iteration of an execution of algorithm $\APA$ specified as follows:
\begin{itemize}
\item Nodes $v\in\HS$ have input $p_v^r$.
\item For nodes $x\in[n]\setminus\HS$ such that $\CB$ with dealer $x$ outputs $\Delta_{v,x}^r\neq \bot$ at $v\in \HS$, $v$ receives $p^r_x$ from $x$.
\end{itemize}
We consider the vector $L$ of values used to compute the value $\Delta_v^r$ in iteration $r$ of $\CPS$. For $w\in[n]$, $\Delta^r_{v,w}=p_w^r-p^r_v$. Hence, $L$ can be obtained by shifting the non-$\bot$ input values used in the above execution of $\APA$ by $-p_v^r$. This implies that in both computations, $v$ assigns the input corresponding to node $w$ to the same position $i_w$ in the sorted vector of non-$\bot$ inputs. Hence, the indices of discarded inputs in the list $L$ remains the same in both of these executions as well. Let $\ell=|L|$ denote the length of $L$ and $b:=n-\ell$ the number of $\bot$ values. Without loss of generality, we assume $L$ to be sorted in ascending order, and that parties are sorted in ascending order by size of their pulse times $\vec{p}$. Then in the $r$th iteration of $\CPS$, $v$ computes the midpoint $\Delta_v^r$ of interval $I_L$ spanned by the remaining points in $L$ as 
\begin{align*}
\Delta_v^r&=\frac{L_{\ell-f+b}+L_{f-b+1}}{2}\\
&=\frac{(p^{r}_{\ell-f+b}-p_v^r)+(p^{r}_{f-b+1}-p_v^r)}{2}\\
&=\frac{p^{r}_{\ell-f+b}+p^{r}_{f-b+1}}{2}-p_v^r.
\end{align*}
By comparison, the above execution of $\APA$ computes the midpoint as $(p^{r}_{\ell-f+b}+p^{r}_{f-b+1})/2=M_v$ for every node $v\in\HS$. Thus, both executions compute the same midpoint, up to a shift of $p_v^r$. That is, $\CPS$ computes the vector of midpoints $\Delta_v^r=M_v-p_v^r$ for every node $v\in\HS$. As $\APA$ satisfies the $1/2$-consistency condition of approximate agreement, we know that the above execution of $\APA$ computes midpoints $\vec{M}$ satisfying $\|\vec{M}\|\leq\|p^r\|/2$. It follows that $\|\vec{\Delta}^r+\vec{p}^r\|=\|\vec{M}\|\leq\|\vec{p}^r\|/2.$ This proves the second statement of the lemma.

To prove the first statement, we use the validity condition of approximate agreement. We obtain for all $v\in\HS$ that $\min_{w\in\HS} \{p^r_w\} \leq M_v\leq \max_{w\in\HS} \{p^r_w\}$. Since $\vec{\Delta}^r=\vec{M}-\vec{p}^r$, we obtain for all $v\in\HS$ that 
\begin{align*}
-\|p^r\|&=\min_{w\in\HS}\{p^r_w\}-\max_{w\in\HS}\{p^r_w\}\leq M_v-p^r_v=\Delta_v^r\\
&\leq \max_{w\in\HS}\{p_w^r\}-\min_{w\in\HS} \{p^r_w\}= \|p^r\|.
\end{align*}

For the general case of $\delta>0$, we note that the list of values received by $v\in \HS$ can be obtained from the one of an execution with $\delta=0$ by adding to each received value a shift between $0$ and $\delta$. Thus, our analysis for the case of $\delta=0$ implies that
\begin{align*}
\frac{p^{r}_{\ell-f+b}+p^{r}_{f-b+1}}{2}-p_v^r\le \Delta_v^r&\le \frac{p^{r}_{\ell-f+b}+\delta+p^{r}_{f-b+1}+\delta}{2}-p_v^r\\
&=\frac{p^{r}_{\ell-f+b}+p^{r}_{f-b+1}}{2}-p_v^r+\delta.
\end{align*}
We conclude that $-\|\vec{p}^r\|\le \Delta_v^r\le \|\vec{p}^r\|+\delta$, showing the first claim of the lemma. Moreover, $\Delta_v^r+p^r_v\in[M_v,M_v+\delta]$ and hence, $\|\Delta^r+p^r\|\leq\|\vec{M}\|+\delta\leq\|p^r\|/2+\delta$.
\end{proof}
An immediate consequence of \Cref{lemma:apa_sim} is that if the hypothesis of our induction holds, i.e., $\|\vec{p}^r\|\le S$, then the computed shifts $\vec{\Delta}^r$ are feasible, in the sense that $v\in \HS$ the waiting statement at the end of the main loop of Algorithm~$\CPS$ refers to a local time that is larger then the local time when all instances of crusader broadcast for round $r$ have terminated at $v$, so that $v$ can compute $\Delta_v^r$; a proof is given in~\Cref{app:proofs}.
\begin{corollary}\label{cor:pulses} Fix $r\in \mathbb{N}$. Suppose that $T\geq (\vartheta^2+\vartheta+1) S+(\vartheta+1)d-2u$ and $\|\vec{p}^r\|\leq S$. Moreover, denote as $\tau^r_v$ the time at which $v\in \HS$ finalizes the computation of $\Delta^r_v$. Then for all $v\in\HS$, $H_v(p_v^r)+\Delta^r_v+T\geq H_v(\tau^r_v)$.
\end{corollary}
\begin{proof}
First, we note that Algorithm~$\TCB$ terminates at the latest at local time $H_v(p_v^r)+(\vartheta+1)d-2u+(\vartheta^2+\vartheta)S$. Hence, $H_v(\tau_v^r)\le H_v(p_v^r)+(\vartheta+1)d-2u+(\vartheta^2+\vartheta)S$. Moreover, due to the condition $\|\vec{p}^r\|\leq S$, we can apply \Cref{lemma:deltabf,lemma:deltabh}, which prove that the preconditions of \Cref{lemma:apa_sim} hold for iteration $r$. We apply the precondition $\|\vec{p}^r\|\le S$ and the first statement of \Cref{lemma:apa_sim}, which together yields that for all $v\in\HS$, $-S\le -\|\vec{p}^r\|\leq \Delta_v^r$. Thus,
\begin{align*}
&H_v(p_v^r)+\Delta^r_v+T\\
\geq\,& H_v(p_v^r)-S+T\\
\geq \,&H_v(\tau_v^r)-((\vartheta+1)d-2u+(\vartheta^2+\vartheta+1)S)+T\\
\geq \,& H_v(\tau_v^r).\qedhere
\end{align*}
\end{proof}

Similarly to~\cite[Ch.~10, Lem.~10.7]{Lec21}, we can now prove that the induction steps succeeds, i.e., that if $\|\vec{p}^r\|\le S$, then also $\|\vec{p}^{r+1}\|\le S$; the details are given in~\Cref{app:proofs}.
\begin{lemma}\label{lemma:induction} Suppose that $T\geq (\vartheta^2+\vartheta+1) S+(\vartheta+1)d-2u$, $S\geq(2(2\vartheta-1)\delta+2(\vartheta-1)T)/(2-\vartheta)$, and $\|\vec{p}^r\|\leq S$. Then
\begin{itemize}
\item $(T-S)/\vartheta\leq\min_{w\in\HS}\{p_{w}^{r+1}\}-\min_{w\in\HS}\{p_{w}^{r}\}\leq T+S +\delta$ and
\item $\|\vec{p}^{r+1}\|\leq S$.
\end{itemize}
\end{lemma}
\begin{proof}[Proof Sketch.]
\Cref{lemma:apa_sim} allows us to interpret the new pulse times as the result of performing an approximate agreement step with error $\delta$ on the previous pulse times (up to adding $T$ and clock drift). The first statement follows from the resulting bound of $-S\le \|\vec{p}^r\|\le \Delta_v^r\le \|\vec{p}^r\|+\delta\le S+\delta$. The second statement is shown by using that approximate agreement reduces the ``old'' error carried over from the previous iteration by $\|\vec{p}^r\|/2$, which compensates for the ``new'' error due to the measurement error of $\delta$ and clock drift.
\end{proof}

The proof of the main theorem is analogous to~\cite[Ch.~10, Thm.~10.9]{Lec21}; we defer it to \Cref{app:proofs}.
\begin{theorem}\label{thm:sync} Suppose that $4-\vartheta+\vartheta^2-3\vartheta^3>0$ and 
\begin{align*}
T\geq\frac{(\vartheta^2+\vartheta+1)2(2\vartheta-1)(2u+(\vartheta^2-1)d)}{4-\vartheta+\vartheta^2-3\vartheta^3}+(\vartheta+1)d-2u\in O(d).
\end{align*}
Define $S$ as
\begin{align*}
S:=\frac{2(2\vartheta-1)(u+(\vartheta-1)d)+2(\vartheta-1)T}{4-2\vartheta-\vartheta^2}\in O(u+(1-1/\vartheta)T),
\end{align*}
and suppose that $\max_{v\in \HS}\{H_v(0)\}\leq S$. Then Algorithm $\CPS$ is an $(\lceil n/2\rceil-1)$-secure clock synchronizaton protocol with skew $S$, minimum period $\minP\geq (T-(\vartheta+1)S)/\vartheta$, and maximum period $\maxP\leq T+3S$.
\end{theorem}
\begin{proof}[Proof Sketch.]
$S$ and $T$ are chosen in accordance with the prerequisites of \Cref{lemma:induction}, from which the claims readily follow. As $T\in \Omega(S)$ and $S\in \Omega((\vartheta-1)T$, this is only feasible if $\vartheta-1$ is small enough. 
\end{proof}
By checking for which values of $\vartheta>1$ the polynomial in $\vartheta$ in the preconditions of the theorem is positive, we arrive at \Cref{cor:sync}.
\sync*
\section{Lower Bound on the Skew}\label{sec:lower}
Algorithm~$\CPS$ exploits signatures in order to prevent faulty nodes from equivocating about their own perception of time, i.e., when they broadcast. However, echoing signatures comes not only at the expense of another message delay (and thus timing uncertainty $u$), but also at the expense of adding indirection to communication relevant for the algorithm. As a result, the timing delay uncertainty determining its skew bound is actually $\tilde{u}$, the uncertainty on communication links involving a faulty party, which might be larger than $u$.

In this section, we prove that this restriction is inherent. Concretely, we show that no algorithm can achieve a skew smaller than $2\tilde{u}/3$, even if $u=0$. Our lower bound is established by the standard technique of manipulating hardware clocks in a way that is hidden by adjusting message delays, enabling us to construct three executions that are indistinguishable to honest nodes. Two executions are indistinguishable to an honest node if the sequence of received messages as well as the local reception times are identical; in this case, the algorithm must send the same messages at the same local times and generate pulses at the same time.\footnote{For a randomized algorithm, this holds only after fixing the randomness. However, the strategy of the adversary is independent of the algorithm, which allows us to prove the result for deterministic algorithms and infer the general statement by Yao's principle.} The hardware clock skew of $2\tilde{u}/3$ we build up then translates to a skew of at least $2\tilde{u}/3$ between pulses in at least one of the three executions.

In our setting, this approach faces the technical challenge that it is not sufficient to define the executions such that hardware clocks and message delays conform to the model, followed by proving their indistinguishability. In addition, we must prove that the adversary can obtain the required knowledge to determine its strategy upfront and have faulty nodes send the required messages to maintain indistinguishability in time.

We now set up to formally prove the lower bound. Fix an arbitrary deterministic pulse synchronization algorithm~$\mathcal{A}$. For convenience of notation, w.l.o.g.\ we assume that no two messages are received at the exact same local time; the general case can be covered by breaking ties lexicographically. Also w.l.o.g., we can assume $n=3$; the general case follows from a simple simulation argument we provide later. In order to distinguish the three executions we construct, we use superscripts $i\in [3]$. For notational convenience, both execution and node indices are always taken modulo~$3$, which also exposes the symmetry of the construction.

Throughout this section, all the execution triples $(\Ex^i)_{i\in [3]}$ we consider share the following properties $P$:
\begin{itemize}
\item $\HS^i=[3]\setminus\{i\}$.
\item Messages sent between honest parties have delay $d$.
\item All other messages have delay $d-\tilde{u}$.
\item $H^i_{i+1}(t)=t$.
\item $H^i_{i+2}(t)=\vartheta t$ for $t\leq 2\tilde{u}/(3(\vartheta-1))$ and $H^i_{i+2}(t)=t+2\tilde{u}/3$ for $t\geq 2\tilde{u}/(3(\vartheta-1))$.
\end{itemize}
In particular, hardware clocks of honest nodes never deviate by more than $2\tilde{u}/3$ from real time.

We construct our triple of executions by induction over the number of messages $k$ sent by faulty nodes, where indistinguishability holds before the maximum local reception time of such a message. The base case is trivial; the following lemma performs the step. 
\begin{lemma}\label{lemma:lbinduction1} Let $n=3$, $k\in\mathbb{N}$, and $\Pi$ be a $1$-secure protocol for pulse synchronization with skew $S$. Suppose there exist executions $\Ex^1,\Ex^2,\Ex^3$ of $\Pi$ satisfying $P$ and the following properties.
\begin{itemize}
\item Faulty nodes send a total of $k$ messages in $\Ex^1$, $\Ex^2$, and $\Ex^3$.
\item Let $h^*$ be the maximum local time at which a message from a faulty node is received (or $0$ if faulty nodes send no messages). Honest node $i$ cannot distinguish $\Ex^{i+1}$ and $\Ex^{i+2}$ until local time $h^*$.
\end{itemize}
If there exists a node $i$ that can distinguish $\Ex^{i+1}$ and $\Ex^{i+2}$, then there exist executions $\TEx^1$, $\TEx^2$, and $\TEx^3$ with the same properties, except that faulty nodes send an additional message which is received at a local time $h>h^*$. Moreover, for $i\in [3]$, nodes $i+1$ and $i+2$ cannot distinguish $\TEx^i$ from $\Ex^i$ before time $h$.
\end{lemma}
\begin{proof} 
Let $h$ denote the minimum local time (over all executions) at which a node $i$ can distinguish $\Ex^{i+1}$ and $\Ex^{i+2}$. By the prerequisites and the assumption that no two messages are received at the same time, there exists an honest message $m$ that $i$ receives in exactly one of these executions at local time $h>h^*$. Observe that the requirements on $\TEx^1,\TEx^2,\TEx^3$ fully specify these executions save for the additional message that is sent by some faulty party. We rule that the additional message is $m$, which is sent by the same (there faulty) node with the same local reception time at $i$. We will need to establish that $\TEx^1,\TEx^2,\TEx^3$ are well-defined, in that faulty nodes always learn all signatures required for the messages they send. However, let us first establish the indistinguishability statements, as these are needed to show well-definedness.

We begin by showing that for each $j\in [3]$, executions $\TEx^{j+1}$ and $\Ex^{j+1}$ are indistinguishable before time $h$. An analogous statement holds for $\TEx^{j+2}$ and $\Ex^{j+2}$. Toward a contradiction, suppose that there exists a local time $\tilde{h}<h$ at which some node $j$ can distinguish $\TEx^{j+1},\Ex^{j+1}$. Without loss of generality, we let $\tilde{h}$ denote the minimal such time. This implies that $j$ receives a message $\tilde{m}$ at local time $\tilde{h}$ in execution $\TEx^{j+1}$, which it does not receive in $\Ex^{j+1}$ (or vice versa).

We analyze two cases:
\begin{itemize}
\item $\tilde{m}$ is sent by an honest node $j'\neq j$: As $d> 2\tilde{u}/3$, honest parties' local clocks are less than $d$ apart at all times. As messages sent between honest parties have delay $d$, thus $j'$ must have sent $\tilde{m}$ at some local time less than $\tilde{h}$. Because $\TEx^{j+1}$ and $\Ex^{j+1}$ are indistinguishable before time $\tilde{h}$, $j'$ sends $\tilde{m}$ in both $\Ex^{j+1}$ and $\TEx^{j+1}$. Hence, $\tilde{m}$ is received at the same local time by $j$ in both $\Ex^{j+1}$ and $\TEx^{j+1}$. This is a contradiction.
\item $\tilde{m}$ is sent by an faulty node $j'\neq j$: As $\tilde{m}$ is received by $j$ before local time $h$, $\tilde{m}\neq m$. It follows that $j'$ sends $\tilde{m}$ in both executions. Hence it is received at the same local time by $j$ in both $\Ex^{j+1}$ and $\TEx^{j+1}$. This is a contradiction.
\end{itemize}
Since $\Ex^{j+1}$ and $\Ex^{j+2}$ are indistinguishable to $j$ before local time $h$, it follows that $\TEx^{j+1}$ and $\Ex^{j+2}$ are also indistinguishable before local time $h$ to $j$. As, analogously, $\Ex^{j+2}$ and $\TEx^{j+2}$ are indistinguishable to $j$ before local time $h$, $\TEx^{j+1}$ and $\TEx^{j+2}$ are also indistinguishable to $j$ before local time $h$. Finally, by construction and the assumption that only one message is received at any given time, $\TEx^{j+1}$ and $\TEx^{j+2}$ are also indistinguishable to $j$ at local time $h$.

It remains to show that the behaviour of faulty nodes in executions $\TEx^{1},\TEx^{2},\TEx^{3}$ is also well-defined.
Suppose that faulty node $i\in [3]$ sends $\tilde{m}$ in $\TEx^i$. Let $\bar{h}\leq h$ be the local time when $\tilde{m}$ is received (by an honest node) in $\TEx^i$. Denote as $t^i\ge \bar{h}-2\tilde{u}/3$ the time at which $\tilde{m}$ is received in $\TEx^i$. By definition, $i$ sends $\tilde{m}$ at time $t^i-d+\tilde{u}\ge \bar{h}-d+\tilde{u}/3$ in $\TEx^i$. It is sufficient to show that $i$ receives every message $m'$ on which $\tilde{m}$ depends by this time. 

Consider such a message $m'$ and suppose that $i+1\neq i$ is its sender; the case that $i+2$ is the sender is treated analogously. We distinguish two cases based on which node receives $\tilde{m}$.
\begin{itemize}
\item $i+2$ receives $\tilde{m}$: Since $H^{i+1}_{i+2}(\bar{h})\ge \bar{h}$, and $i$ sends $\tilde{m}$ no later than time $\bar{h}-d$ in $\TEx^{i+1}$. Thus, $m'$ is received by $i$ in $\TEx^{i+1}$ no later than at time $\bar{h}-d$, which corresponds to local time at most $h^{i+1}:=H^{i+1}_i(\bar{h}-d)\leq\bar{h}-d+2\tilde{u}/3<\bar{h}$. By indistinguishability of executions before local time $h\ge \bar{h}$, $i$ receives $m'$ by the same local time $h^{i+1}$ in $\TEx^{i+2}$, and accordingly no later than (real) time $h^{i+1}$. Hence, $i+1$ sends $m'$ in $\TEx^{i+2}$ by time $h^{i+1}-d$ and local time $h^{i+2}:=H^{i+2}_{i+1}(h^{i+1}-d)\leq h^{i+1}-d+2\tilde{u}/3<\bar{h}$.
By indistinguishability of $\TEx^{i}$ and $\TEx^{i+2}$ before local time $h$, it follows that $i+1$ sends $m'$ in $\TEx^{i}$ by local time (and also time) $h^{i+2}$. Hence, in $\TEx^{i}$, $i$ receives $m'$ by time 
\begin{align*}
h^{i+2}+d-\tilde{u}&\leq h^{i+1}-\frac{\tilde{u}}{3}\leq \bar{h}-d+\frac{\tilde{u}}{3}.
\end{align*}
\item $i+1$ receives $\tilde{m}$: Since $H^{i+2}_{i+1}(\bar{h})\ge\bar{h}$, we have that $i$ sends $\tilde{m}$ by time $\bar{h}-d$ in $\TEx^{i+2}$. Thus, $m'$ is received by $i$ in $\TEx^{i+2}$ by time $\bar{h}-d$, and hence sent by $i+1$ by time $\bar{h}-2d$. This corresponds to local time $h^{i+2}:=(H^{i+2}_{i+1})(\bar{h}-2d)\leq \bar{h}-2d+2\tilde{u}/3<\bar{h}$. By indistinguishability of executions before local time $h\ge \bar{h}$, $i+1$ sends $m'$ by the same local time and time $h^{i+2}$ in $\TEx^i$. Hence, $i$ receives $m'$ at time at most 
\begin{align*}
h^{i+2}+d-\tilde{u}&\leq \bar{h}-d-\frac{\tilde{u}}{3}<\bar{h}-d+\frac{\tilde{u}}{3}.\qedhere
\end{align*}
\end{itemize}
\end{proof}
Using the \Cref{lemma:lbinduction1} inductively, we construct executions that are indistinguishable to honest nodes at all times. A full proof of \Cref{lemma:lbinduction2} is provided in the appendix.
\begin{lemma}\label{lemma:lbinduction2} Let $n=3$ and let $\Pi$ be a $1$-secure protocol for pulse synchronization with skew $S$. Then there exist executions $\Ex^1,\Ex^2,\Ex^3$ of $\Pi$ that satisfy $P$ and where node $i\in [3]$ cannot distinguish $\Ex_{i+1}$ and $\Ex_{i+2}$.
\end{lemma}
\begin{proof}[Proof Sketch.]
Inductive application of \Cref{lemma:lbinduction1} yields a series of triples of executions with an increasing number of messages sent by faulty nodes, which cannot be distinguished by correctness up to the largest local time when such a message is received. There are two cases: The induction halts, because the indistinguishability holds at all times; in this case we are done. Otherwise, the constructed executions share indistinguishable prefixes with those that are constructed later. As only finitely many messages are sent in finite time, these prefixes must become arbitrarily long. Hence, they define limit executions satisfying the required properties.
\end{proof}

Equipped with these executions, we are in the position to prove the claimed lower bound.
\lb*
\begin{proof} We show the statement for the special case of $n=3,t=1$. For general $n$, we can reduce the argument to the case of $n=3,t=1$ as follows. Assume the existence of an $\lceil n/3\rceil$-secure protocol $\Pi$ for pulse synchronization with skew $S<2\tilde{u}/3$. Now, partition the set of $n$ nodes into three non-empty subsets $S_1,S_2,S_3$ of size at most $\lceil n/3\rceil$. Then, node $i\in[3]$ simulates the protocol behaviour of nodes in $S_i$ in $\Pi$ and outputs the pulse times of the lexicographically first node in $S_i$. By assumption, this yields a $1$-secure protocol for pulse synchronization for $n=3$ with skew $S$. 

Thus, let $\Pi$ be a $1$-secure protocol for pulse synchronization, and assume for now that $\Pi$ is deterministic. By Lemma \ref{lemma:lbinduction2}, there exist executions $\Ex^{1},\Ex^{2},\Ex^{3}$ with the following properties:
\begin{itemize}
\item $\HS^i=[3]\setminus\{i\}$.
\item $H^i_{i+1}(t)=t$.
\item $H^i_{i+2}(t)=\vartheta\cdot t$ for $t\leq 2\tilde{u}/(3(\vartheta-1))$ and $H^i_{i+2}(t)=t+2\tilde{u}/3$ for $t\geq 2\tilde{u}/(3(\vartheta-1))$.
\item $\Ex_{i+1}$ and $\Ex_{i+2}$ are indistinguishable to node $i$.
\end{itemize}
Recall that $\Pi$ must guarantee some minimum period $\minP>0$. We define
$r:=\lceil\tilde{u}/\minP(\vartheta-1)\rceil+1$, such that for all $i\in[3]$, $\min\{p_{i+1}^{i,r},p_{i+2}^{i,r}\}\geq 2\tilde{u}/(3(\vartheta-1))$, where $p_v^{i,r}$ denotes the $r$th pulse time of honest node $v$ in execution $\Ex^i$.
Recall that for all times $t$, $H^{i+2}_i(t)=t$ and for times $t\geq p_i^{i+1,r}$, $H^{i+1}_i(t)= t+2\tilde{u}/3$. By indistinguishability of executions $\Ex^{i+1}$ and $\Ex^{i+2}$ for node $i$, $H_i^{i+1}(p_i^{i+1,r})=H_i^{i+2}(p_i^{i+2,r})$. Hence,
\begin{align*}
p_i^{i+1,r}&=(H_i^{i+1})^{-1}(H_i^{i+1}(p_i^{i+1,r}))\\
&=(H_i^{i+1})^{-1}(H_i^{i+2}(p_i^{i+2,r}))\\
&=(H_i^{i+1})^{-1}(p_i^{i+2,r})=p_i^{i+2,r}-\frac{2\tilde{u}}{3}.
\end{align*}
We conclude that
\begin{align*}
3S&\geq (p_{i+1}^{i,r}-p_{i+2}^{i,r})+(p_{i+2}^{i+1,r}-p_{i}^{i+1,r})+(p_{i}^{i+2,r}-p_{i+1}^{i+2,r})\\
&=(p_i^{i+2,r}-p_{i}^{i+1,r})+(p_{i+1}^{i,r}-p_{i+1}^{i+2,r})+(p_{i+2}^{i+1,r}-p_{i+2}^{i,r})\\
&=2\tilde{u}.
\end{align*}
This implies that $S\geq 2\tilde{u}/3$.

It remains cover the case that $\Pi$ is randomized. To this end, we interpret $\Pi$ as random variable evaluating to a deterministic protocol (as a result of fixing the randomness of the nodes). Independently of $\Pi$, the adversary uniformly at random picks node $i\in [3]$ to corrupt. These choices determine executions $\Ex^i$, $i\in [3]$, as above. The adversary then lets $i$ behave such that $\Ex^i$ is realized.

We claim that, while the adversary might not be able to learn the randomness of nodes $i+1$ and $i+2$, it is not required to do so. The messages node $i$ sends are those it would send as a correct node in $\Ex^{i+1}$ and $\Ex^{i+2}$, respectively. It receives the same messages in $\Ex^i$, just at different times. However, from the reception times it can compute the corresponding local reception times at $i$ in $\Ex^{i+1}$ and $\Ex^{i+2}$, based on the known message delays and local times of the respective senders in the respective executions. Hence, it can simply simulate two copies of $\Pi$ at $i$ with the randomness of $i$, to which it feeds the received messages with the local times at which they are received in $\Ex^{i+1}$ and $\Ex^{i+2}$, respectively. As we have shown that all messages required in the adversary's simulation are received early enough to produce all appropriate messages to be sent in time, we conclude that the adversary can indeed realize $\Ex^i$.

Denote by $p_v^{i,r}(\pi)$ the $r$th pulse of $v$ in the execution $\Ex^i$ of deterministic protocol $\pi$. Using independence, we get that
\begin{align*}
3\cdot\mathbb{E}[S]&=3\cdot\sum_{i=1}^3 P[i \mbox{ is corrupted}]\cdot\mathbb{E}[S\,|\,i \mbox{ is corrupted}]\\
&=\sum_{i=1}^3 \mathbb{E}[S\,|\,i \mbox{ is corrupted}]\\
&\ge \sum_{i=1}^3 \sum_{\pi} P[\Pi=\pi]\cdot (p_{i+1}^{i,r}(\pi)-p_{i+2}^{i,r}(\pi))\\
&=\sum_{\pi} P[X=\pi]\cdot \sum_{i=1}^3 p_{i+1}^{i,r}(\pi)-p_{i+2}^{i,r}(\pi))\\
&\ge \sum_{\pi} P[X=\pi]\cdot 2\tilde{u}=2\tilde{u}.\qedhere
\end{align*}
\end{proof}

\bibliographystyle{plain} 
\bibliography{references.bib}

\begin{thebibliography}{10}

\bibitem{ACDNPS19}
Ittai Abraham, T.-H.~Hubert Chan, Danny Dolev, Kartik Nayak, Rafael Pass, Ling
  Ren, and Elaine Shi.
\newblock Communication complexity of byzantine agreement, revisited.
\newblock In {\em Principles of Distributed Computing (PODC)}, pages 317--326,
  2019.

\bibitem{ADDNR19}
Ittai Abraham, Srinivas Devadas, Danny Dolev, Kartik Nayak, and Ling Ren.
\newblock Synchronous byzantine agreement with expected $o(1)$ rounds, expected
  $o(n^2)$ communication, and optimal resilience.
\newblock In {\em Financial Cryptography}, pages 320--334, 2019.

\bibitem{A85}
Baruch Awerbuch.
\newblock Complexity of network synchronization.
\newblock {\em Journal of the ACM (JACM)}, 32(4):804---823, 1985.

\bibitem{BW01}
Sa\^{a}d Biaz and Jennifer~L. Welch.
\newblock Closed form bounds for clock synchronization under simple uncertainty
  assumptions.
\newblock {\em Information Processing Letters (IPL)}, 80(3):151---157, 2001.

\bibitem{BKLL20}
Erica Blum, Jonathan Katz, Chen-Da Liu-Zhang, and Julian Loss.
\newblock Asynchronous byzantine agreement with subquadratic communication.
\newblock In {\em Theory of Cryptography Conference (TCC)}, pages 353--380,
  2020.

\bibitem{BLR19}
Johannes Bund, Christoph Lenzen, and Will Rosenbaum.
\newblock Fault tolerant gradient clock synchronization.
\newblock In {\em ACM Symposium on Principles of Distributed Computing (PODC)},
  pages 357--365, 2019.

\bibitem{CKS00}
Christian Cachin, Klaus Kursawe, and Victor Shoup.
\newblock Random oracles in constantipole: practical asynchronous byzantine
  agreement using cryptography (extended abstract).
\newblock In {\em Principles of Distributed Computing (PODC)}, volume 123--132,
  2000.

\bibitem{CR93}
Ran Canetti and Tal Rabin.
\newblock Fast asynchronous byzantine agreement with optimal resilience.
\newblock In {\em ACM Symposium on Theory of Computing (STOC)}, pages 42--51,
  1993.

\bibitem{CPS20}
T.-H.~Hubert Chan, Rafael Pass, and Elaine Shi:.
\newblock Sublinear-round byzantine agreement under corrupt majority.
\newblock In {\em Public Key Cryptography (PKC)}, pages 246--265, 2020.

\bibitem{CKS20}
Shir Cohen, Idit Keidar, and Alexander Spiegelman.
\newblock Not a coincidence: Sub-quadratic asynchronous byzantine agreement
  whp.
\newblock In {\em International Symposium on Distributed Computing (DISC)},
  pages 25:1--25:17, 2020.

\bibitem{D82}
Danny Dolev.
\newblock The byzantine generals strike again.
\newblock {\em Journal of Algorithms}, 3(1):14--30, 1982.

\bibitem{Dolev82}
Danny Dolev.
\newblock The byzantine generals strike again.
\newblock {\em Journal of Algorithms}, 3(1):14--30, 1982.

\bibitem{DHS84}
Danny Dolev, Joe Halpern, and H.~Raymond Strong.
\newblock On the possibility and impossibility of achieving clock
  synchronization.
\newblock In {\em ACM Symposium on Theory of Computing (STOC)}, pages
  504---511, 1984.

\bibitem{Lec21}
Danny Dolev and Christoph Lenzen.
\newblock Clock synchronization and adversarial fault tolerance.
\newblock Online lecture notes, 2021.
\newblock
  https://www.mpi-inf.mpg.de/departments/algorithms-complexity/teaching/summer21/clock-synchronization-and-adversarial-fault-tolerance.

\bibitem{DLPSW86}
Danny Dolev, Nancy~A. Lynch, Shlomit~S. Pinter, Eugene~W. Stark, and William~E.
  Weihl.
\newblock Reaching approximate agreement in the presence of faults.
\newblock {\em Journal of the ACM (JACM)}, 33(3):499---516, 1986.

\bibitem{DS83}
Danny Dolev and Raymond Strong.
\newblock Authenticated algorithms for byzantine agreement.
\newblock {\em SIAM Journal on Computing}, 12(4):656--666, 1983.

\bibitem{DY83}
Danny Dolev and Andrew Chi-Chih Yao.
\newblock On the securituy of public key protocols.
\newblock {\em IEEE Transactions on Information Theory}, 29(2):198--207, 1983.

\bibitem{FL04}
Rui Fan and Nancy Lynch.
\newblock Gradient clock synchronization.
\newblock In {\em ACM Symposium on Principles of Distributed Computing (PODC)},
  pages 320---327, 2004.

\bibitem{FLP85}
Michael Fischer, Nanch Lynch, and Robert Patterson.
\newblock Impossibility of distributed consensus with one faulty process.
\newblock {\em Journal of the ACM}, 32(2):374--382, 1985.

\bibitem{FLM85}
Michael~J. Fischer, Nancy~A. Lynch, and Michael Merritt.
\newblock Easy impossibility proofs for distributed consensus problems.
\newblock In {\em ACM Symposium on Principles of Distributed Computing (PODC)},
  pages 59---70, 1985.

\bibitem{HSSD84}
Joseph~Y. Halpern, Barbara Simons, Ray Strong, and Danny Dolev.
\newblock Fault-tolerant clock synchronization.
\newblock In {\em ACM Symposium on Principles of Distributed Computing (PODC)},
  pages 89---102, 1984.

\bibitem{KL18}
Pankaj Khanchandani and Christoph Lenzen.
\newblock Self-stabilizing byzantine clock synchronization with optimal
  precision.
\newblock {\em Theory of Computing Systems (TOCS)}, 63:261--305, 2018.

\bibitem{LSP82}
Leslie Lamport, Robert~E. Shostak, and Marshall~C. Pease.
\newblock The byzantine generals problem.
\newblock {\em ACM Transactions on Programming Languages and Systems (TOPLAS)},
  4(3):382--401, 1982.

\bibitem{LLW10}
Christoph Lenzen, Thomas Locher, and Roger Wattenhofer.
\newblock Tight bounds for clock synchronization.
\newblock {\em Journal of the ACM (JACM)}, 57(2), 2010.

\bibitem{LL84}
Jennifer Lundelius and Nancy Lynch.
\newblock A new fault-tolerant algorithm for clock synchronization.
\newblock In {\em ACM Symposium on Principles of Distributed Computing (PODC)},
  pages 75---88, 1984.

\bibitem{Micali17}
Silvio Micali.
\newblock Very simple and efficient byzantine agreement.
\newblock In {\em Innovations in Theoretical Computer Science Conference
  (ITCS)}, 2017.

\bibitem{MXCSS16}
Andrew Miller, Yu~Xia, Kyle Croman, Elaine Shi, and Dawn Song.
\newblock The honey badger of bft protocols.
\newblock In {\em ACM SIGSAC Conference on Computer and Communications Security
  (CCS)}, pages 31--42, 2016.

\bibitem{ST85}
T.~K. Srikanth and Sam Toueg.
\newblock Optimal clock synchronization.
\newblock In {\em Principles of Distributed Computing (PODC)}, pages 71--86,
  1985.

\bibitem{WXDS20}
Jun Wan, Hanshen Xiao, Srinivas Devadas, and Elaine Shi.
\newblock Round-efficient byzantine broadcast under strongly adaptive and
  majority corruptions.
\newblock In {\em Theory of Cryptography Conference (TCC)}, pages 412--456,
  2020.

\end{thebibliography}

\appendix

\section{Further Related Work}\label{app:further}

\medskip\noindent\textbf{Cryptography in Distributed Consensus.} Cryptography makes it possible to circumvent the famous $f<n/3$ bound due to Lamport, Shostak, and Pease~\cite{LSP82} on the number of tolerable corruptions in a synchronous protocol solving the consensus problem. This was first demonstrated in the seminal work of Dolev and Strong~\cite{DS83} who gave an authenticated algorithm for Byzantine broadcast tolerating any $f<n-1$ corrupted parties. Early works in the area typically use cryptography, e.g., signatures, as structure-free objects, in line with their symbolic treatment in the Dolev-Yao model~\cite{DY83}. Cryptography has also been a staple tool to facilitate randomized consensus protocols, which are inherently required in the asynchronous setting due to the FLP result~\cite{FLP85}. Here, cryptography is typically used in the form of more advanced primitives such as secret sharing~\cite{CR93}, multi-party computation~\cite{BKLL20}, or threshold signatures~\cite{CKS00}, which can be used to agree on a random coin flip. Other advanced forms of cryptography have been (more recently) used to optimize the round and communication complexity of consensus protocols. Verifiable random functions produce efficiently verifiable, yet unpredictable values that can be used to efficiently elect a random leader (or a random subcommittee) in a consensus protocol~\cite{Micali17,ACDNPS19,CPS20,CKS20}. Threshold encryption allows to agree on a ciphertext of unknown content which can be forcibly decrypted by a sufficient number of honest parties. This has served as a useful tool for guaranteeing liveness in asynchronous consensus protocols~\cite{MXCSS16}. Recent work~\cite{WXDS20} has also shown the feasibility of round efficient consensus with a strongly adaptive adversary capable of after-the-fact-removal of messages sent by honest parties using time-locked puzzles. Lastly, we mention the work of Abraham et al.~\cite{ADDNR19} who give randomized consensus protocols from primitives such as the above and also provide a signature-based variant of the Srikanth-Toueg pulse synchronizer~\cite{ST85}. We point out, however, that their algorithm, while tolerating the optimal corruption fraction of $\lceil n/2\rceil -1$ Byzantine parties, exhibits a skew of up to $d$, whereas we achieve $O(\tilde{u}+\vartheta-1)d)$. Likewise,~\cite{HSSD84} achieves skew larger than $d$ for the related task of clock synchronization using essentially the same scheme.

\medskip\noindent\textbf{Fault-tolerant Clock Synchronization in General Networks.} Little work has directly addressed fault-tolerant synchronization in non-complete networks. For a known topology, without signatures (node) connectivity $(2f+1)$ is necessary and sufficient to simulate full connectivity in the presence of up to $f$ faults~\cite{D82}, lifting results for consensus under full connectivity to general networks. Analogously, $(2f+1)$-connectivity can be leveraged to simulate full connectivity with suitable timing information to do the same for clock synchronization, by taking the median of clock offset estimation obtained via $2f+1$ node-disjoint paths. This comes at the expense of such estimates being subject to the maximum over the used paths of the uncertainty accumulated along the path, which is justified by~\cite{BW01}, which shows for any pair of nodes a lower bound proportional to the length of the shortest path between them in the fault-free setting. Although we are not aware of a formal proof of this claim, $(2f+1)$-connectivity is also necessary for synchronization, since a majority of faulty nodes on a node cut of the network allows for the faulty majority to claim arbitrarily large clock offsets between the (thus effectively disconnected) parts of the network.

In light of the lower bound of $\Omega(uD)$ on the worst-case skew\footnote{We state all bounds here for uniform link delays and uncertainties, but all bounds can be generalized to he heterogenous case.} in fault-free networks of diameter $D$~\cite{BW01}, Fan and Lynch proposed to study the task of gradient clock synchronization, in which the goal is to keep the local skew, i.e., the skew between neighbors in the network, small~\cite{FL04}. Matching upper and lower bounds of $\Theta(u\log D)$ on the local skew have been proven for the fault-free case~\cite{LLW10}. A fault-tolerant generalization of the algorithm achieves the same asymptotic skew bounds in the presence of $f$ faults, provided that the (arbitrary connected) base network is augmented by copying nodes and links $\Theta(f)$ times~\cite{BLR19}. The existing gradient clock synchronization algorithms require no knowledge of the topology.

In the setting with signatures, $(f+1)$-connectivity is trivially necessary and sufficient to simulate full connectivity of the network. This allows to carry over the existing algorithms with skew $\Theta(d)$ to this setting~\cite{ADDNR19,HSSD84}, where $d$ becomes the worst-case end-to-end delay across the network after deleting the faulty nodes (\cite{HSSD84} in fact states the general result). Our algorithm can be translated to any known $(f+1)$-connected network in the same way, where $\tilde{u}$ and $d$ are replaced by the maximum end-to-end delay and uncertainty over all paths used to simulate full connectivity. Note that in addition to making sure that communication has stable latency on the link level, one needs to balance the length (in terms of overall delay) of the utilized paths in order to keep $\tilde{u}$ much smaller than $d$.

\section{Ommitted Proofs}\label{app:proofs}

\begin{proof}[Proof of~\Cref{lemma:deltabh}] Due to Lemma \ref{lem:correctscb1}, $\Delta_{v,w}^r\neq\bot$. Denote as $t_w$ the time that the dealer $w$ sends its message and let $t_v\in [t_w+d-u,t_w+d]$ denote the time that $v$ receives it. We bound $\Delta_{v,w}^r$ from below as 
\begin{align*}
\Delta_{v,w}^r&=H_v(t_v)-H_v(p_v^r)-d+u-S\\
&\geq t_v-p_v^r-d+u-S\\
&\ge t_w-p_v^r-S\\
&=t_w+p_w^r-p_w^r-p_v^r-S\\
&\geq \frac{H_w(t_w)-H_w(p^r_w)}{\vartheta}+p_w^r-p_v^r-S= p_w^r-p_v^r.
\end{align*}
The upper bound is derived as
\begin{align*}
&\Delta_{v,w}^r\\
=\,&H_v(t_v)-H_v(p_v^r)-d+u-S\\
\leq \,&\vartheta(t_v-p_v^r)-d+u-S\\
\le\,&\vartheta(t_w+d-p_v^r)-d+u-S\\
=\,&p_w^r-p_v^r+u+\vartheta(t_w-p_w^r)+(\vartheta-1)(d+p_w^r-p_v^r)-S\\
\le \,&p_w^r-p_v^r+u+\vartheta(t_w-p_w^r)+(\vartheta-1)d+(\vartheta-2)S\\
\le \,&p_w^r-p_v^r+u+\vartheta(H_w(t_w)-H_w(p_w^r))+(\vartheta-1)d+(\vartheta-2)S\\
=\,& p_w^r-p_v^r+u+(\vartheta-1) d +(\vartheta^2+\vartheta-2) S<p_w^r-p_v^r+\delta.\qedhere
\end{align*}
\end{proof}

\begin{proof}[Proof of \Cref{lemma:induction}.]
We begin by noting that due to \Cref{cor:pulses} (which we can apply due to the lemma conditions), $\vec{p}^{r+1}$ is well-defined. Next, we apply \Cref{lemma:deltabf,lemma:deltabh} to show that the preconditions of \Cref{lemma:apa_sim} hold for pulse $r$. Thus, by the first statement of~\Cref{lemma:apa_sim}, for all $v\in \HS$ we have that
\begin{align*}
-S\le -\|\vec{p}^r\|\leq \Delta_v^r\le \|\vec{p}_r\|+\delta\le S+\delta.
\end{align*}
Consider $v:=\arg \min_{w\in\HS}\{p^{r+1}_w\}$. Then 
\begin{align*}
\min_{w\in\HS}\{p_{w}^{r+1}\}-\min_{w\in\HS}\{p_{w}^{r}\}&\geq p_v^{r+1}-p_v^r\\
&\geq \frac{H_v(p_v)^{r+1}-H_v(p_v^{r})}{\vartheta}\\
&\geq \frac{T+\Delta^r_v}{\vartheta}\\
&\geq \frac{T-S}{\vartheta}.
\end{align*}
This establishes the lower bound in item 1. For the upper bound, consider $v:=\arg \min_{w\in\HS}\{p^{r}_w\}$. We get
\begin{align*}
\min_{w\in\HS}\{p_{w}^{r+1}\}-\min_{w\in\HS}\{p_{w}^{r}\}&\leq p^{r+1}_v-p^r_v\\
&\leq H_v(p_v^{r+1})-H_v(p_v^{r})\\
&=T+\Delta_v^r\\
&\leq T+S+\delta.
\end{align*}

For the second statement of the lemma, fix arbitrary $v,w\in\HS$ and assume w.l.o.g.\ that $p_w^{r+1}\geq p_v^{r+1}$. We have that
\begin{align*}
&p_w^{r+1}-p_v^{r+1}\\
=\,&p_w^r+p_w^{r+1}-p_w^r-(p_v^r+p_v^{r+1}-p_v^r)\\
\le \,&p_w^r+H_w(p_w^{r+1})-H_w(p_w^r)-\left(p_v^r-\frac{H_v(p_v^{r+1})-H_v(p_v^r)}{\vartheta}\right)\\
=\,&p_w^{r}+\Delta_w^r-(p_v^{r} +\Delta_v^r)+\left(1-\frac{1}{\vartheta}\right)(T+\Delta_v^r)\\
\leq \,&p_w^{r}+\Delta_w^r-(p_v^{r} +\Delta_v^r)+\left(1-\frac{1}{\vartheta}\right)(T+S+\delta)\\
\leq\,&\frac{\|\vec{p}^r\|}{2}+\delta+\left(1-\frac{1}{\vartheta}\right)(T+S+\delta)\\
\leq\,&\frac{S}{2}+\delta+\left(1-\frac{1}{\vartheta}\right)(T+S+\delta)\le S,
\end{align*}
where the second inequality uses the already established bound $\Delta_v^r\leq S+\delta$, the third inequality follows from the second statement of \Cref{lemma:apa_sim}, the second to last inequality applies the precondition $\|\vec{p}^r\|\leq S$, and the final inequality holds due to the precondition on $S$. Since $v,w\in \HS$ were arbitrary, we conclude that $\|\vec{p}_{r+1}\|\le S$.
\end{proof}

\begin{proof}[Proof of~\Cref{thm:sync}.]
We begin by observing that 
\begin{align*}
S&=\frac{2(2\vartheta-1)(2u+(\vartheta^2-1)d)+2(\vartheta-1)T}{2-\vartheta+\vartheta^2-\vartheta^3}\\
&=\frac{2(2\vartheta-1)\delta+2(\vartheta-1)T}{2-\vartheta}>\delta
\end{align*}
and 
\begin{align*}
T&\geq\frac{(\vartheta^2+\vartheta+1)2(2\vartheta-1)(2u+(\vartheta^2-1)d)}{4-\vartheta+\vartheta^2-3\vartheta^3}+(\vartheta+1)d-2u\\
&=(\vartheta^2+\vartheta+1) S+(\vartheta+1)d-2u.
\end{align*}
We prove the theorem by induction on the pulse number $r$, where the induction hypothesis is that $\|\vec{p}^r\|\leq S$ and if $r\neq 0$ also
\begin{align*}
\min_{w\in\HS}\{p_{w}^{r+1}\}-\max_{w\in\HS}\{p_{w}^{r}\}&\geq \frac{T-(\vartheta+1)S}{\vartheta}\quad\mbox{and}\\
\max_{w\in\HS}\{p_{w}^{r+1}\}-\min_{w\in\HS}\{p_{w}^{r}\}&\leq T+3S.
\end{align*}
The base case holds due to the condition of the theorem. For the step case, assume that for $r\in\mathbb{N}$, $\|\vec{p}^r\|\leq S$. We invoke \Cref{lemma:induction}, yielding that $\|\vec{p}^{r+1}\|\leq S$ and that $(T-S)/\vartheta\leq\min_{w\in\HS}\{p_{w}^{r+1}\}-\min_{w\in\HS}\{p_{w}^{r}\}\leq T+S +\delta$. Thus,
\begin{align*}
\min_{w\in\HS}\{p_{w}^{r+1}\}-\max_{w\in\HS}\{p_{w}^{r}\}&\geq\min_{w\in\HS}\{p_{w}^{r+1}\}-\min_{w\in\HS}\{p_{w}^{r}\}+S\\
&\geq {T-(\vartheta+1)S}{\vartheta}
\end{align*}
and 
\begin{align*}
\max_{w\in\HS}\{p_{w}^{r+1}\}-\min_{w\in\HS}\{p_{w}^{r}\}&\leq\min_{w\in\HS}\{p_{w}^{r+1}\}-\min_{w\in\HS}\{p_{w}^{r}\}+S\\
&\leq T+2S+\delta<T+3S.\qedhere
\end{align*}
\end{proof}

\pprotocol{Algorithm $\CB$}{A $2$-round synchronous algorithm for Crusader Broadcast.}{fig:cb}{t!}{
\begin{itemize}
\item The dealer $v$ sends $(b_v,\sig{b_v}{v})$ to all nodes.
\item Let $(b,\sigma)$ be the value received from the dealer. Send $(b,\sigma)$ to all nodes.
\item Let $(b_w,\sigma_w)$ be the message received from node $w$. Output $\bot$ if either of the following occurs:
\begin{itemize}
\item There are nodes $w_1\neq w_2$ such that $\ver(pk_v,\sigma_{w_1},b_{w_1})=\ver(pk_v,\sigma_{w_2},b_{w_2})=1$, but $b_{w_1}\neq b_{w_2}$.
\item $\ver(pk_v,\sigma,b)=0.$
\end{itemize}
Otherwise, output $b$.
\end{itemize}
}

\begin{proof}[Proof of \Cref{lemma:lbinduction2}.]
We show the statement by inductively constructing triples of executions $(\Ex_k^1,\Ex^k_2,\Ex_k^3)$ for $k\in\mathbb{N}$ satisfying $P$ and the following properties:
\begin{itemize}
\item Faulty nodes send a total of $k$ messages in $\Ex^1_k$, $\Ex^2_k$, and $\Ex^3_k$.
\item Let $h_k$ be the maximum local time at which a message from a faulty node is received (or $0$ if faulty nodes send no messages). Honest node $i$ cannot distinguish $\Ex^{i+1}_k$, $\Ex^{i+2}_k$, $\Ex^{i+1}_{k+1}$, and $\Ex^{i+2}_{k+1}$ until time $h_k$.
\item Honest node $i$ cannot distinguish $\Ex^{i+1}_k$ and $\Ex^{i+2}_k$ until time $h_k$.
\end{itemize}
Note that these conditions imply all preconditions of Lemma~\ref{lemma:lbinduction1}, except for the existence of a node $i\in[3]$ that can distinguish executions $\Ex^{i+1}$ and $\Ex^{i+2}$ after local time $h_k$. We begin by noting that we can define $(\Ex_0^1,\Ex^0_2,\Ex_0^3)$ with the necessary properties by having faulty nodes send no messages, setting message delays to $d$, and defining the functions $H^i_j$ for all $i,j\in[3]$ as stated above. This covers the base case of our induction. 

Our induction stops at finite index $k_0$ if there does not exist a node $i\in[3]$ which can distinguish executions $\Ex_{k_0}^{i+1}$ and $\Ex_{k_0}^{i+2}$. For all $k$ not of this form, observe that if $\Ex^1_k$, $\Ex^2_k$, and $\Ex^3_k$ with the above properties exist, then all preconditions of Lemma~\ref{lemma:lbinduction1} are met. Hence, we can apply Lemma~\ref{lemma:lbinduction1} for the step case of our induction, which yields $\Ex_{k+1}^{1},\Ex_{k+1}^2$ and $\Ex_{k+1}^{2}$ which satisfy the properties stated above. This concludes the induction.

To conclude the proof, we distinguish two cases:
\begin{itemize}
\item The induction halts at finite index $k_0$: Then $\Ex_{k_0}^1$, $\Ex_{k_0}^2$, and $\Ex_{k_0}^3$ satisfy the requirements of the lemma.
\item The induction does not halt: We note that for any two indices $k_1<k_2$ and $i\in[3]$, executions $\Ex_{k_1}^i$ and $\Ex_{k_2}^i$ are indistinguishable up to local time $h_{k_1}$ in the view of nodes $i+1$ and $i+2$. Recall that this means that nodes $i+1$ and $i+2$ send and receive all messages in these executions at the same local times (before local time $h_{k_1}$). This also means that these nodes send and receive the same messages at the same local times before time $t_{k_1}:=\min_{j\in\{i+1,i+2\}}(H^{i}_{j})^{-1}(h_{k_1})\geq h_{k_1}-2\tilde{u}/3$. Also, this implies that the faulty node $i$ in these executions sends and receives the same messages before time $t_{k_1}-d+\tilde{u}= h_{k_1}-d+\tilde{u}/3$. Thus, we see that for any $k$, executions $\Ex_{k}^i,\Ex_{k+1}^i,...$ are identical until time $t_k-d+\tilde{u}/3$. By our model assumptions, nodes send a finite amount of messages over any finite period of time. This implies that $t_k$ grows unbounded as $k$ increases. Hence, there is a well-defined limit execution $\Ex^i$ which for all $k\in\mathbb{N}$ is identical to $\Ex_{k}^i$ before time $t_k$. These executions satisfy the claim of the lemma.\qedhere
\end{itemize}
\end{proof}

\end{document}